\documentclass[10pt,reqno]{amsart}
\usepackage{amsmath,amsthm,amssymb,url,enumerate}
\usepackage{graphicx}
\usepackage{afterpage}
\usepackage{nameref}
\usepackage[all]{xy}
\usepackage[colorlinks=true]{hyperref}
\usepackage[width=5.5in,height=8.5in,centering]{geometry}
\usepackage{algorithm}
\usepackage{array}
\makeatletter
\let\orgdescriptionlabel\descriptionlabel
\renewcommand*{\descriptionlabel}[1]{%
  \let\orglabel\label
  \let\label\@gobble
  \phantomsection
  \edef\@currentlabel{#1}%
  \let\label\orglabel
  \orgdescriptionlabel{#1}%
}
\makeatother

\theoremstyle{theorem}
\newtheorem{theorem}{Theorem}
\theoremstyle{plain}
\numberwithin{theorem}{section}
\newtheorem{problem}[theorem]{Problem}
\newtheorem{proposition}[theorem]{Proposition}
\newtheorem{lemma}[theorem]{Lemma}

\theoremstyle{definition}
\newtheorem{definition}[theorem]{Definition}
\theoremstyle{remark}

  {\end{list}}

  {\end{list}}

\def\Gcal{{\mathcal G}}

\def\Ucal{{\mathcal U}}

\newcommand{\FF}{\mathbb{F}}

\newcommand{\QQ}{\mathbb{Q}}
\newcommand{\RR}{\mathbb{R}}
\newcommand{\ZZ}{\mathbb{Z}}

\def \bfv{{\mathbf v}}

\newcommand{\Disc}{\operatorname{Disc}}

\newcommand{\MOD}[1]{~(\textup{mod}~#1)}
\renewcommand{\pmod}{\MOD}

\title{Provably Weak Instances of Ring-LWE}
\author{Yara Elias \and Kristin E. Lauter \and Ekin Ozman \and Katherine E. Stange}
\address{Yara Elias:
Department of Mathematics and Statistics, McGill University, Montreal, Quebec, Canada}
\email{yara.elias@mail.mcgill.ca}
\address{Kristin E. Lauter:Microsoft Research, One Microsoft Way, Redmond, WA 98052}
\email{klauter@microsoft.com}
\address{Ekin Ozman: Department of Mathematics, Faculty of Arts and Science, Bogazici University, 34342, Bebek-Istanbul, Turkey}
\email{ekin.ozman@boun.edu.tr}
\address{Katherine E. Stange: Department of Mathematics, University of Colorado, Campux Box 395, Boulder, Colorado 80309-0395}
\email{kstange@math.colorado.edu}
\date{\today}
\begin{document}


\maketitle \begin{abstract}  The \emph{ring and polynomial learning with errors} problems (Ring-LWE and Poly-LWE) have been proposed as hard problems to form the basis for cryptosystems, and various security reductions to hard lattice problems have been presented.  So far these problems have been stated for general (number) rings but have only been closely examined for cyclotomic number rings. In this paper, we state and examine the Ring-LWE problem for general number rings and demonstrate {\it provably weak instances} of the Decision Ring-LWE problem.  We construct an explicit family of number fields for which we have an efficient attack.  We demonstrate the attack in both theory and practice, providing code and running times for the attack.   The attack runs in time linear in $q$, where $q$ is the modulus.

Our attack is based on the attack on Poly-LWE which was presented in~\cite{EHL}.
We extend the EHL-attack to apply to a larger class of number fields, and show how it applies to attack Ring-LWE for a heuristically large class of fields. Certain Ring-LWE instances can be transformed into Poly-LWE instances without distorting the error too much, and thus provide the first weak instances of the Ring-LWE problem.
We also provide additional examples of fields which are vulnerable to our attacks on Poly-LWE, including power-of-$2$ cyclotomic fields, presented using the minimal polynomial of $\zeta_{2^n} \pm 1$.
\end{abstract}

\section{Introduction}

Lattice-based cryptography has become a very hot research topic recently with the emergence of new applications to homomorphic encryption.
The hardness of the Ring-LWE problem was related to
various well-known hard lattice problems~\cite{Regev,MR09,MR04,LPR10,BLPRS}, and the hardness of the Poly-LWE problem was reduced to Ring-LWE in~\cite{LPR10,DD12}.
The hardness of the  Poly-LWE problem is used as the basis of security for numerous cryptosystems, including~\cite{BV11,BGV11,GHS11}.
The hardness of Ring-LWE was also shown \cite{SS11} to form a basis for the proof of security of a variant of NTRU \cite{HPS,1363}.

In~\cite{EHL}, the first weaknesses in the Poly-LWE problem were discovered for classes of number fields satisfying certain properties.  In addition, a list of properties of number fields were identified which are sufficient
to guarantee a reduction between the Ring-LWE and the Poly-LWE  problems, and a search-to-decision reduction for Ring-LWE.  Unfortunately, in ~\cite{EHL}, no number fields were found which satisfied both the conditions for the attack and for the reductions.  Thus ~\cite{EHL} produced only examples of number fields which were weak instances for Poly-LWE.

The contributions of this paper at a high level are as follows: In Section~\ref{sec:attack} we strengthen and extend the attacks presented in~\cite{EHL} in several significant ways.
In Section \ref{sec:moving}, most importantly, we show how the attacks can be applied also to the Ring-LWE problem. In Section~\ref{sec:provable}, we construct an explicit family of number fields for which we have an efficient attack on the Decision Ring-LWE Problem.
This represents the first successful attacks on the Decision Ring-LWE problem for number fields with special properties.
For Galois number fields, we also know that an attack on the decision problem gives an attack on the search version of Ring-LWE
(\cite{EHL}).
In addition, in Section~\ref{sec:exampleattack},  we present the first successful implementation of the EHL attack at cryptographic sizes and attack both Ring-LWE and Poly-LWE instances.  For example for $n=1024$ and $q=2^{31}-1$, the attack runs in about 13 hours.
Code for the attack is given in Appendix \ref{sec:code}. In Section \ref{sec:heuristic} we give a more general construction of number fields such that heuristically a large percentage of them will be vulnerable to the attacks on Ring-LWE.

In more detail, we consider rings of integers in number fields $K=\QQ[x]/(f(x))$ of degree $n$, modulo a large prime number $q$, and we give attacks on Poly-LWE which work when  $f(x)$ has a root of small order modulo $q$.
The possibility of such an attack was mentioned in~\cite{EHL} but not explored further.
In Sections \ref{sec:SmallSet} and \ref{sec:SmallError}, we give \emph{two} algorithms for this attack, and in Sections \ref{sec:ex1} and \ref{sec:ex2} we give many examples of number fields and moduli, some of cryptographic size, which are vulnerable to this attack.
 The most significant consequence of the attack is the construction of the number fields which are weak for the Ring-LWE problem (Section \ref{sec:heuristic}).

To understand the vulnerability of Ring-LWE to these attacks, we state and examine the Ring-LWE problem for general number rings and demonstrate {\it provably weak instances} of Ring-LWE.  We demonstrate the attack in both theory and practice for an explicit family of number fields, providing code and running times for the attack.   The attack runs in time linear in $q$, where $q$ is the modulus.  The essential point is that Ring-LWE instances can be mapped into Poly-LWE instances, and if the map does not distort the error too much, then the instances may be vulnerable to attacks on Poly-LWE.
The distortion is governed by the spectral norm of the map, and we compute the spectral norm for the explicit family we construct in Section~\ref{sec:provable} and analyze when the attack will succeed.
For the provably weak family which we construct, the feasibility of the attack depends on the ratio of $\sqrt{q}/n$.  We prove that the attack succeeds when $\sqrt{q}/n$ is above a certain bound, but in practice we find that we can attack instances where the ratio is almost $100$ times smaller than that bound.
Even for Ring-LWE examples which are not taken from the provably weak family, we were able to attack in practice relatively generic instances of number fields  where the spectral norm was small enough (see Section~\ref{sec:exampleattack}).

 We investigate cyclotomic fields (even $2$-power cyclotomic fields) given by an alternate minimal polynomial, which are weak instances of Poly-LWE for that choice of polynomial basis. Section \ref{sec:ex2} contains numerous examples of $2$-power cyclotomic fields which are vulnerable to attack when instantiated using an alternative polynomial basis, thus showing the heavy dependence in the hardness of these lattice-based problems on the choice of polynomial basis.
In addition, we analyze the case of cyclotomic fields to understand their potential vulnerability to these lines of attack
and we explain why cyclotomic fields are immune to attacks based on roots of small order (Section \ref{sec:cyc}).
Finally, we provide code in the form of simple routines in SAGE to implement the attacks and algorithms given in this paper and demonstrate successful attacks with running times (Section \ref{sec:exampleattack}).

As a consequence of our results, one can conclude that the hardness of Ring-LWE is both {\it dependent on special properties of the number field} and {\it sensitive to the particular choice of $q$}, and some choices may be significantly weaker than others.  In addition, for applications to cryptography, since our attacks on Poly-LWE run in time roughly $O(q)$ and may be applicable to a wide range of fields, including even $2$-power cyclotomic fields with a bad choice of polynomial basis, these attacks should be taken into consideration when selecting parameters for Poly-LWE-based systems such as~\cite{BV11,BGV11} and other variants. For many important applications to homomorphic encryption (see for example~\cite{GLN,BLN}), these attacks will not be relevant, since the modulus $q$ is chosen large enough to allow for significant error growth in computation, and would typically be of size $128$ bits up to $512$ bits.  For that range, the attacks presented in this paper would not run.  However, in other applications of  Ring-LWE to key exchange for the TLS protocol~\cite{BCNS},   parameters for achieving $128$-bit security are suggested where $n=2^{10}$ and $q=2^{32}-1$, with $\sigma \approx 3$, and these parameters would certainly be vulnerable to our attacks for weak choices of fields and $q$.

{\bf Acknowledgements.}  The authors are indebted to the organizers of the research conference Women in Numbers 3 (Rachel Pries, Ling Long and the fourth author), as well as to the Banff International Research Station, for bringing together this collaboration.  They would also like to thank Hao Chen for his careful reading of the manuscript, correcting typos, and for providing an improved argument in Section \ref{sec:hs}.
Finally, the authors thank Martin Albrecht for help with Sage.

\section{Background on Poly-LWE}

Let $f(x)$ be a monic irreducible polynomial in $\ZZ[x]$ of degree $n$, and let $q$ be a prime such that $f(x)$ factors completely modulo $q$.  Let $P = \ZZ[x]/f(x)$ and let $P_q = P/qP = \FF_q[x]/f(x)$.  Let $\sigma \in \RR^{>0}$.
The uniform distribution on $P \simeq \ZZ^n$ will be denoted $\mathcal{U}$.  By \emph{Gaussian distribution of parameter $\sigma$} we refer to a discrete Gaussian distribution of mean $0$ and variance $\sigma^2$ on $P$, spherical with respect to the power basis.  This will be denoted $\mathcal{G}_\sigma$.
It is important to our analysis that we assume that in practice, elements are sampled from Gaussians of parameter $\sigma$ truncated at width $2 \sigma$.

There are two standard Poly-LWE problems.  Our attack solves the \emph{decision} variant, but it also provides information about the secret.

\begin{problem}[Decision Poly-LWE Problem] \label{PLWED}
Let $s(x) \in P$ be a secret.  The \emph{decision Poly-LWE problem} is to distinguish, with non-negligible advantage, between the same number of independent samples in two distributions on $P\times P$.  The first consists of samples of the form $(a(x), b(x) := a(x) s(x) + e(x) )$ where $e(x)$ is drawn from a discrete Gaussian distribution of parameter $\sigma$, and $a(x)$ is uniformly random.  The second consists of uniformly random and independent samples from $P \times P$.

\end{problem}

\begin{problem}[Search Poly-LWE Problem] \label{PLWES}
         Let $s(x) \in P$ be a secret. The \emph{search Poly-LWE problem}, is to discover $s$ given access to arbitrarily many independent samples of the form $(a(x), b(x) := a(x)s(x) + e(x))$ where $e(x)$ is drawn from a Discrete Gaussian of parameter $\sigma$, and $a(x)$ is uniformly random.
\end{problem}

\noindent
The polynomial $s(x)$ is called the \emph{secret} and the polynomials $e_i(x)$ are called the \emph{errors}.

\subsection{Parameter selection}
\label{sec:parameters}

The selection of parameters for security is not yet a well-explored topic.  Generally parameter recommendations for Poly-LWE and Ring-LWE are just based on the recommendations for general LWE, ignoring the extra ring structure e.g. \cite{PG,RV,BCNS}.  Sample concrete parameter choices  have been suggested, where $w$ is the width of the Gaussian error distribution  (precisely, $w = \sqrt{2\pi} \sigma$):

\begin{enumerate}
        \item $P_{LP1} = (n,q,w) = (192, 4093, 8.87)$, $P_{LP2} = (256, 4093, 8.35)$, $P_{LP3} = (320, 4093, 8.00)$ for low, medium and high security, recommended by Lindner and Peikert in \cite{LP};
        \item $P_{GF} = (n,q,w) = (512, 12289,12.18)$ for high security used in \cite{GF};
        \item $P_{BCNS} = (n,q,w) = (1024, 2^{31}-1, 3.192)$ suggested in \cite{BCNS} for the TLS protocol.  Here, $q = 2^{32}-1$ was actually suggested but it is not prime.  Here, the authors remark that $q$ is taken to be large for correctness but could potentially be decreased.
\end{enumerate}

\section{Attacks on Poly-LWE}
\label{sec:attack}

The attack we are concerned with is quite simple.  It proceeds in four stages:
\begin{enumerate}
        \item Transfer the problem to $\FF_q$ via a ring homomorphism $\phi: P_q \rightarrow \FF_q$.
        \item Loop through guesses for the possible images $\phi(s(x))$ of the secret.
        \item Obtain the values $\phi(e_i(x))$ under the assumption that the guess at hand is correct.
        \item Examine the distribution of the $\phi(e_i(x))$ to determine if it is Gaussian or uniform.
\end{enumerate}

If $f$ is assumed to have a root $\alpha \equiv 1 \mod{q}$ or $\alpha$ of small order modulo $q$, then this attack is due to Eisentraeger-Hallgren-Lauter \cite{EHL}.

The first part is to transfer the problem to $\FF_q$.  Write
        $f(x) = \prod_{i=1}^n (x-\alpha_i)$
for the factorization of $f(x)$ over $\FF_q$ which is possible by assumption.  By the Chinese remainder theorem, if $f$ has no double roots, then
\[
        P_q \simeq \prod_{i=1}^n \FF_q[x]/(x-\alpha_i) \simeq \FF_q^n
\]
There are $n$ ring homomorphisms
\[
        \phi: P_q \rightarrow \FF_q[x]/(x-\alpha_i) \simeq \FF_q, \quad g(x) \mapsto g(\alpha_i).
\]
Fix one of these, by specifying a root $\alpha = \alpha_i$ of $f(x)$ in $\FF_q$.  Apply the homomorphism to the coordinates of the $\ell$ samples $(a_i(x), b_i(x))$, obtaining
$(a_i(\alpha), b_i(\alpha))_{i=1,\ldots, \ell}$.

Next,  loop through all $g \in \FF_q$.  Each value $g$ is to be considered a guess for the value of $s(\alpha)$.  For each guess $g$, assuming that it is a correct guess  and $g=s(\alpha)$, then
\[
        e_i(\alpha) = b_i(\alpha) - a_i(\alpha)g = b_i(\alpha) - a_i(\alpha)s(\alpha).
\]
In the case that the samples were LWE samples and the guess was correct, then this produces a collection $(e_i(\alpha))$ of images of errors chosen according to some distribution.
If the distributions $\phi(\mathcal{U})$ and $\phi(\mathcal{G}_{\sigma})$ are distinguishable, then we can determine whether
the distribution was uniform or Gaussian.
Note that $\phi(\mathcal{U})$ will of course be uniform on $\FF_q$.  If our guess is incorrect, or if the samples are not LWE samples, then the distribution will appear uniform.

Therefore, after looping through all guesses, if all the distributions appeared uniform, then conclude that the samples were not LWE samples; whereas if one of the guesses worked for all samples and always yielded an error distribution which  appeared Gaussian, assume that particular $g$ was a correct guess.
In the latter case this also yields one piece of information about the secret:  $g=s(\alpha) \mod q$.

The attack \emph{will} succeed whenever
\begin{enumerate}
        \item $q$ is small enough to allow looping through $\FF_q$,
        \item $\phi(\mathcal{U})$ and $\phi(\mathcal{G}_{\sigma})$ are distinguishable.
\end{enumerate}

Our analysis hinges on the difficulty of distinguishing $\phi(\mathcal{U})$ from $\phi(\mathcal{G}_{\sigma})$, as a function of the parameters $\sigma$, $n$, $\ell$, $q$, and $f$.  Distinguishability becomes easier when $\sigma$ is smaller (so $\mathcal{U}$ and $\mathcal{G}_{\sigma}$ are farther apart to begin with), $n$ is smaller and $q$ is larger (since then less information is lost in the map $\phi$), and $\ell$ is larger (since there are more samples to test the distributions).  The dependence on $f$ comes primarily as a function of its roots $\alpha_i$ modulo $q$, which may have properties that make distinguishing easier.

Ideally, for higher security, one will choose parameters that make distinguishing nearly impossible, i.e. such that $\phi(\mathcal{G}_{\sigma})$ appears very close to uniform modulo $q$.\\

{\bf Example. (\cite{EHL})} We illustrate the attack in the simplest case $\alpha=1$. Assume $f(1) \equiv 0 \hbox{ mod } q$, and consider the distinguishability of the two distributions $\phi(\Ucal)$ and $\phi(\Gcal_\sigma)$.
Given $(a_i(x), b_i(x))$, make a guess $g\in \FF_q$ for the value of $s(1)$ and compute $b_i(1)-g\cdot a_i(1)$.
If $b_i $ is uniform, then $b_i(1)-g \cdot a_i(1)$ is uniform  for all $g$.
If $b_i=a_i s +e_i$, then there is a guess $g$ for which $b_i(1)-ga_i(1)=e_i(1)$ where $e_i(x)=\sum_{j=1}^n e_{ij}x^j$ and $g=s(1)$.
Since $e_i(1) = \sum_{j=1}^n e_{ij}$, where $e_{ij}$ are chosen from $\mathcal{G}_{\sigma}$, it follows that $e_i(1)$ are sampled from $\Gcal_{\sqrt{n}\sigma}$ where $n\sigma^2<< q$. The attack can be described loosely as follows: for each sample, test each guess $g$ in $\FF_q$ to see if $b_i(1)-g \cdot a_i(1)$ is small modulo $q$, and only keep those guesses which pass the test. Repeat with the next sample and continue to keep only the guesses which pass.

\subsection{Attack based on a small set of error values modulo $q$}
\label{sec:SmallSet}

In this section, we assume that there exists a root $\alpha$ of $f$ such that $\alpha$ has small order $r$ modulo $q$, that is $\alpha^r \equiv 1 \mod{q}$. Then
\begin{equation}
        \label{eqn:ealpha}
        e(\alpha)= \sum_{i=0}^{n-1} e_i \alpha^i = (e_0 + e_r+e_{2r}+ \cdots) + \alpha(e_1+e_{r+1}+\cdots) + \cdots + \alpha^{r-1}(e_{r-1}+e_{2r-1}+\cdots).
\end{equation}
If $r$ is small enough, then $e(\alpha)$ takes on only a small number of values modulo $q$.  If so, then we can efficiently distinguish whether a value modulo $q$ belongs to that subset.

Let $S$ be the set of possible values of $e(\alpha)$ modulo $q$.   We assume for simplicity that $n$ is divisible by $r$.  Then the coefficients $e_j+e_{j+r}+\cdots+ e_{n-r+j}$ of \eqref{eqn:ealpha} fall into a subset of $\ZZ/q\ZZ$ of size at most $4\sigma n/r$.
We sum over $r$ terms,
hence, $|S| = (4\sigma n/r )^r$ residues modulo $q$.  For $r=2$, this becomes $(2n\sigma)^2$.

The attack described below succeeds with high probability if $|S| << q$, that is
\begin{equation*}
        (4\sigma n /r)^r << q.
\end{equation*}

\begin{algorithm}[H]
			\caption{Small set of error values}
			\label{alg:SmallSet}
                        \raggedright
{\bf Input:} A collection of $\ell$ Poly-LWE samples.

{\bf Output:} A guess $g$ for $s(\alpha)$, the value of the secret polynomial at $\alpha$; or else {\bf NOT PLWE}; or \emph{\bf INSUFFICIENT SAMPLES}.

The value {\bf NOT PLWE} indicates that the collection of samples were definitely not Poly-LWE samples.

The value {\bf INSUFFICIENT SAMPLES} indicates that there were not enough samples to determine a single guess $s(\alpha)$.  In this case, the algorithm may be continued on a new set of samples by looping the remaining surviving guesses on the new samples.

\vspace{1em}

Create an ordered list of elements of $S$.

Let $G$ be an empty list.

{\bf for} $g$ from $0$ to $q-1$ {\bf do}

\quad\quad {\bf for} $(a(x), b(x))$ in the collection of samples {\bf do}

\quad \quad\quad\quad {\bf if} $b(\alpha)- g a(\alpha)$ does not equal an element of $S$ {\bf then}

\quad\quad\quad\quad\quad\quad {\bf break} (i.e. begin next value of $g$)

\quad\quad append $g$ to $G$ (note: occurs only if the loop of samples completed without a break)

{\bf if} $G$ is empty {\bf then}

\quad\quad return {\bf NOT PLWE}

{\bf if} $G = \{ g \} $ {\bf then}

\quad\quad return $g$

{\bf if} $\#G > 1$ {\bf then}

\quad return {\bf INSUFFICIENT SAMPLES}

\end{algorithm}

\begin{proposition}
        \label{prop:SmallSet}
Assume that
\begin{equation}
        \label{eqn:SmallSet}
        (4\sigma n /r)^r < q.
\end{equation}
Algorithm~\ref{alg:SmallSet} terminates in time at most $\widetilde{O}(\ell q + n q)$, where the $\widetilde{O}$ notation hides the $log(q)$ factors and the implied constant depends upon $r$.
        Furthermore, if the algorithm returns {\bf NOT PLWE}, then the samples were not valid Poly-LWE samples.
        If it outputs anything other than {\bf NOT PLWE}, then the samples are valid Poly-LWE samples with probability $1-(\frac{\#S}{q})^\ell$.
In particular, this probability tends to $1$ as $\ell$ grows.

\end{proposition}

\begin{proof}
        As discussed above, there are at most $q$ possible values for the elements of $S$ under the assumption \eqref{eqn:SmallSet}.  To compute each one takes $n$ additions per coefficient (of which there are $r$), combined with an additional $r$ multiplications and $r$ additions.  (Here we have assumed the $\alpha^i$ have been computed; this takes $r$ multiplications.)  Each addition or multiplication takes time at most $\log q$.  Therefore, computing $S$ takes time at most $\widetilde{O}(qnr)$.  For sorting, it is best to sort as $S$ is computed; placing each element correctly takes $\log q$ time.

The principal double loop takes time at most $\widetilde{O}(\ell q)$. If $b(\alpha)$ and $a(\alpha)$ are precomputed, then for each guess $g$, the computation of $b(\alpha)- g a(\alpha)$ only costs one multiplication and one subtraction modulo $q$ (i.e. $2\log q$) while it requires only $\log q$ bit comparisons to decide whether this is in the set $S$.

In Step 4, for later samples, only guesses which were successful in the previous samples (i.e. gave a value which was in the set $S$) are considered. For a sample chosen uniformly at random, one expects the number of successful guesses to be roughly $\frac{\#S}{q}$.  Thus for the second sample, we repeat the above test for only $(\#S)$ guesses.  At the $\ell^{th}$ sample, retaining only guesses which were successful for all previous samples, we expect to test only  $(\frac{\#S}{q})^\ell q$ guesses, which very quickly goes to zero.  Hence, if we examine $\ell$ samples, our tolerance for false positives is proportional to $(\frac{\#S}{q})^\ell$.

\end{proof}

\subsection{Attack based on the size of the error values}
\label{sec:SmallError}

In this section, we describe the most general $\phi: P_q \rightarrow \FF_q$ attack on the Poly-LWE problem, one which can be carried out in any situation.  The rub is that the probability of success will be vanishingly small unless we are in a very special situation.  Therefore our analysis actually bolsters the security of Poly-LWE.

Suppose that
$f(\alpha) \equiv 0 \mod q$.
Let $E_i$ be the event that $b_i(\alpha)-ga_i(\alpha) \mod q$ is in the interval  $[-q/4, q/4)$ for some sample $i$ and guess $g$ for $s(\alpha) \mod q$. The main idea is to compare $P(E_i\mid\mathcal{D}=\mathcal{U})$ and $P(E_i\ | \ \mathcal{D} = \mathcal{G}_{\sigma}).$
If $\mathcal{D} = \mathcal{U}$, then $b_i(\alpha)-ga_i(\alpha)$ is random modulo $ q$ for all guesses $g $, that is,
$P(E_i \ | \ \mathcal{D} = \mathcal{U}) =\frac{1}{2}.$
If $\mathcal{D} = \mathcal{G}_{\sigma}$, then $b_i(\alpha)-s(\alpha)a_i(\alpha)=e_i(\alpha) \mod{q}$.
We consider $$e_i(\alpha)=\sum_{j=0}^{n-1} e_{ij} \alpha^j,$$ where $e_{ij}$ is chosen according to the distribution $\mathcal{G}_{\sigma}$ (truncated at $2\sigma$) and distinguish two cases:
\begin{enumerate}
\item $\alpha = \pm 1$ \label{one}
\item $\alpha \neq \pm 1$ and $\alpha$ has small order $r\geq 3$ modulo $ q$ \label{two}
\end{enumerate}

{\bf Case 1} ($\alpha = \pm 1$).

The error $e_i( \alpha )$ is chosen according to the distribution $\mathcal{G}_{ \sigma \sqrt{n} }$ truncated at $2 \sigma \sqrt{n}$.
Hence $$ -2  \sigma \sqrt{n}\leq e_i( \alpha ) \leq 2   \sigma\sqrt{n}.$$
Therefore, assuming that $$2  \sigma \sqrt{n} < \dfrac{q}{4},$$ we obtain
$P(E_i \ | \ \mathcal{D} = \mathcal{G}_{\sigma})=1$ for $g=s(\alpha)$.  Hence $\mathcal{U}$ and $\mathcal{G}_\sigma$ are distinguishable.

{\bf Case 2} ($\alpha \neq \pm 1$ and $\alpha$ has small order $r\geq 3$ modulo $ q$).

The error can be written as
$$e(\alpha)= \sum_{i=0}^{r-1} e_i \alpha^i = (e_{0}+e_{r}+ \cdots) + \alpha(e_1+e_{r+1}+\cdots) + \cdots + \alpha^{r-1}(e_{r-1}+e_{2r-1}+\cdots)$$
where we assume that $n$ is divisible by $r$ for simplicity.
For $j=0, \cdots,r -1,$ we have that $ e_j + e_{ j+r } + \cdots + e_{ j+ n-r }  $ is chosen according to the distribution $ \mathcal{G}_{\sqrt{\frac{n}{r}} \sigma  }$.
As a consequence $e(\alpha)$ is sampled from $\mathcal{G}_{ \bar{\sigma} }$ where $$\bar{\sigma}^2 = \sum_{i=0}^{r-1} \dfrac{n}{r} \sigma^2 \alpha^{2i} = \dfrac{n}{r} \sigma^2 \dfrac{\alpha^{2r}-1}{\alpha^2-1}.$$
Hence $$ -2 \dfrac{\sqrt{n}}{\sqrt{r}} \sigma \dfrac{\sqrt{\alpha^{2r}-1}}{\sqrt{\alpha^2-1}} \leq e( \alpha ) \leq 2  \dfrac{\sqrt{n}}{\sqrt{r}} \sigma \dfrac{\sqrt{\alpha^{2r}-1}}{\sqrt{\alpha^2-1}} .$$
Therefore, assuming that
\begin{align} \label{bound}
2 \dfrac{\sqrt{n}}{\sqrt{r}} \sigma \dfrac{\sqrt{\alpha^{2r}-1}}{\sqrt{\alpha^2-1}} < \dfrac{q}{4},
\end{align}
we obtain
$P(E_i \ | \ \mathcal{D} = \mathcal{G}_{\sigma})=1$
for $g=s(\alpha)$, and uniform and Gaussian are distinguishable.
Note that Hypothesis \eqref{two} implies in particular that $\alpha^r >q$.\\

\begin{algorithm}[H]
			\caption{Small error values}
			\label{alg:SmallError}
                        \raggedright
{\bf Input:} A collection of $\ell$ Poly-LWE samples.

{\bf Output:} A guess $g$ for $s(\alpha)$; or else {\bf NOT PLWE}; or {\bf INSUFFICIENT SAMPLES}.

The output {\bf INSUFFICIENT SAMPLES} indicates that more samples are needed to make a determination.  In this case, the algorithm can be continued by looping through remaining surviving guesses on new samples.

\vspace{1em}

Let $G$ be an empty list.

{\bf for} $g$ from $1$ to $q-1$ {\bf do}

\quad\quad {\bf for} $(a(x), b(x))$ in the collection of samples {\bf do}

\quad \quad\quad\quad {\bf if} the minimal residue $b(\alpha)- g a(\alpha)$ does not lie in $[-q/4,q/4)$ {\bf then}

\quad\quad\quad\quad\quad\quad {\bf break} (i.e. begin next value of $g$)

\quad\quad append $g$ to $G$ (note: occurs only if the loop of samples completed without a break)

{\bf if} $G$ is empty {\bf then}

\quad\quad return {\bf NOT PLWE}

{\bf if} $G = \{ g \} $ {\bf then}

\quad\quad return $g$

{\bf if} $\#G > 1$ {\bf then}

\quad return {\bf INSUFFICIENT SAMPLES}

\end{algorithm}

In each of the two cases, we have given conditions on the size of $\sigma$ under which $\mathcal{U}$ and $\mathcal{G}_\sigma$ are distinguishable and an attack is likely to succeed.  We now elaborate on the algorithm that would be used.

We denote by $\ell$ the number of samples observed.
For each guess $g$ mod $q$, we compute $b_i - g a_i$ for $i=1, \ldots, \ell$.
If there is a guess $g$ mod $q$ for which the event $E_i$ occurs for all $i=1, \ldots,\ell$, then the algorithm returns the guess if it is unique and {\bf INSUFFICIENT SAMPLES} otherwise; the samples are likely valid Poly-LWE samples.
Otherwise, it reports that they are certainly not valid Poly-LWE samples.

\begin{proposition}
        \label{prop:SmallError}
        Assume that we are in one of the following cases:
        \begin{enumerate}
                \item $\alpha = \pm 1$ and \[ 8\sigma \sqrt{n} < q. \]
                \item  $\alpha$ has small order $r\geq 3$ modulo $ q$, and
                        \[
8 \sigma \dfrac{\sqrt{n}}{\sqrt{r}} \dfrac{\sqrt{\alpha^{2r}-1}}{\sqrt{\alpha^2-1}} < q.
\]
       \end{enumerate}

        Then Algorithm~\ref{alg:SmallError} terminates in time at most $\widetilde{O}(\ell q)$, where the implied constant is absolute.
        Furthermore, if the algorithm returns {\bf NOT PLWE}, then the samples were not valid Poly-LWE samples.
        If it outputs anything other than {\bf NOT PLWE}, then the samples are valid Poly-LWE samples with probability at least $1-(\frac{1}{2})^\ell$.
\end{proposition}

\begin{proof}
        The proof is as in Proposition \ref{prop:SmallSet}, without the first few steps.
\end{proof}

We remark that Propositions and Algorithms \ref{prop:SmallSet} and \ref{prop:SmallError} overlap in some cases.  For $\alpha=\pm 1$, Algorithm \ref{alg:SmallError} is more applicable (i.e. more parameter choices are susceptible), while for $\alpha$ of other small orders, Algorithm \ref{alg:SmallSet} is more applicable.

\section{Moving the attack from Poly-LWE to Ring-LWE}
\label{sec:moving}

We use the term Poly-LWE to refer to LWE problems generated by working in a polynomial ring, and reserve the term Ring-LWE for LWE problems generated by working with the canonical embedding of a number field as in \cite{LPR10,LPR2}.
In the previous sections we have expanded upon Eisentr\"ager, Hallgren and Lauter's observation that for certain distributions on certain lattices given by Poly-LWE, the ring structure presents a weakness.  We will now consider whether it is possible to expand that analysis to LWE instances created through Ring-LWE for number fields besides cyclotomic ones.

In particular, the necessary ingredient is that the distribution be such that under the ring homomorphisms of Section \ref{sec:attack}, the image of the errors is a `small' subset of $\ZZ/q\ZZ$, either the error values themselves are small, or they form a small, identifiable subset of $\ZZ/q\ZZ$.  Assuming a spherical Gaussian in the canonical embedding of $R$ or $R^\vee$, we describe a class of number fields for which this weakness occurs.  A similar analysis would apply without the assumption that the distribution is spherical in the canonical embedding.

Here, we setup the key players (a number field and its canonical embedding, etc.) for general number fields so that these definitions specialize to those in \cite{LPR2}.  There are some choices inherent in our setup:  it may be possible to generalize Ring-LWE to number fields in several different ways.  We consider the two most natural ways.

\subsection{The canonical embedding}
\label{sec:embedding}
Let $K$ be a number field of degree $n$ with ring of integers $R$ whose dual is $R^\vee$.  We will embed the field $K$ in $\RR^n$.  Note that our setup is essentially that of \cite{DD12}, rather than \cite{LPR2}, but the difference is notational.

Let $\sigma_1, \ldots, \sigma_n$ be the $n$ embeddings of $K$, ordered 
so that $\sigma_1$ through $\sigma_{s_1}$ are the $s_1$ real embeddings, and the remaining $n-s_1 = 2s_2$ complex embeddings are paired in such a way that $\overline{\sigma_{s_1+k}} = \sigma_{s_1+s_2+k}$ for $k=1, \ldots, s_2$ (i.e. list $s_2$ non-pairwise-conjugate embeddings and then list their conjugates following that).

Define a map $\theta : K \rightarrow \RR^n$ given by
\[
\theta(r) = (\sigma_1(r), \ldots, \sigma_{s_1}(r), Re(\sigma_{s_1+1}(r)), \ldots, Re(\sigma_{s_1 + s_2}(r)), Im(\sigma_{s_1 + 1}(r)), \ldots, Im(\sigma_{s_1+s_2}(r)) ).
\]
The image of $K$ is the $\QQ$-span of $\theta(\omega_i)$ for any basis $\omega_i$ for $K$ over $\QQ$.  This is not the usual Minkowski embedding, but it has the virtues that 1) the codomain is a real, not complex, vector space; and 2) the spherical or elliptical Gaussians used as error distributions in \cite{LPR2} are, in our setup, spherical or elliptical with respect to the usual inner product.   We denote the usual inner product by $\langle \cdot, \cdot \rangle$ and the corresponding length by $|x| = \sqrt{\langle x, x \rangle}$.  It is related to the trace pairing on $K$, i.e. $\langle \theta(r), \theta(s) \rangle = \operatorname{Tr}(r\overline{s})$.

Then $R$ and $R^\vee$ form lattices in $\RR^n$.

\subsection{Spherical Gaussians and error distributions}

We define a \emph{Ring-LWE error distribution} to be a spherical Gaussian distribution in $\RR^n$.  That is, for a parameter $\sigma > 0$, define the \emph{continuous Gaussian distribution function} $D_\sigma: \RR^n \rightarrow (0,1]$ by
\[
        D_\sigma(x) := (\sqrt{2\pi}\sigma)^{-n}\operatorname{exp}\left( - |x|^2 / (2\sigma^2) \right).
\]

This gives a distribution $\Psi$ on $K \otimes \RR$, via the isomorphism $\theta$ to $\RR^n$.  By approximating $K \otimes \RR$ by $K$ to sufficient precision, this gives a distribution on $K$.

From this distribution we can generate the \emph{Ring-LWE error distribution} on $R$, respectively $R^\vee$, by taking a valid discretization $\lfloor \Psi \rceil_{R}$, respectively $\lfloor \Psi \rceil_{R^\vee}$, in the sense of \cite{LPR2}.  Now we have at hand a lattice, $R$, respectively $R^\vee$, and a distribution on that lattice.  The parameters (particularly $\sigma$) are generally advised to be chosen so that this instance of LWE is secure against general attacks on LWE (which do not depend on the extra structure endowed by the number theory).

\subsection{The Ring-LWE problems}

Write
$R_q := R/qR$ and $R_q^\vee = R^\vee/qR^\vee$.
The standard Ring-LWE problems are as follows, where $K$ is taken to be a cyclotomic field \cite{LPR10,LPR2}.

\begin{definition}[Ring-LWE Average-Case Decision \cite{LPR10}]
        Let $s \in R_q^\vee$ be a secret.  The \emph{average-case decision Ring-LWE problem}, is to distinguish with non-negligible advantage between the same number of independent samples in two distributions on $R_q \times R_q^\vee$.  The first consists of samples of the form $(a, b:=as +e )$ where $e$ is drawn from $\chi := \lfloor \Psi \rceil_{R^\vee}$ and $a$ is uniformly random, and the second consists of uniformly random and independent samples from $R_q \times R_q^\vee$.
\end{definition}

\begin{definition}[Ring-LWE Search \cite{LPR10}]
        Let $s \in R_q^\vee$ be a secret.  The \emph{search Ring-LWE problem}, is to discover $s$ given access to arbitrarily many independent samples of the form $(a, b:=as +e )$ where $e$ is drawn from $\chi := \lfloor \Psi \rceil_{R^\vee}$ and $a$ is uniformly random.
\end{definition}

In proposing general number field Ring-LWE, one of two avenues may be taken:
\begin{enumerate}
        \item preserve these definitions exactly as they are stated, or
        \item eliminate the duals, i.e. replace every instance of $R^\vee$ with $R$ in the definitions above.
\end{enumerate}

To distinguish these two possible definitions, we will refer to \emph{dual Ring-LWE} and \emph{non-dual Ring-LWE}.
Lyubashevsky, Peikert and Regev remark that for cyclotomic fields, dual and non-dual Ring-LWE lead to computationally equivalent problems \cite[Section 3.3]{LPR10}.  They go on to say that over cyclotomics, for implementation and efficiency reasons, dual Ring-LWE is superior.

Generalising dual Ring-LWE to general number fields is the most naive approach, but it presents the problem that working with the dual in a general number field may be difficult.  Still, it is possible there are families of accessible number fields for which this may be the desired avenue.

We will analyse the effect of the Poly-LWE vulnerability on both of these candidate definitions.  In fact, the analysis will highlight some potential differences in their security, already hinted at in the discussion in \cite[Section 3.3]{LPR10}.

\subsection{Isomorphisms from $\theta(R)$ to a polynomial ring}

Suppose $K$ is a \emph{monogenic number field}, meaning that
$R$ is isomorphic to a polynomial ring $P = \ZZ[X]/f(X)$ for some monic irreducible polynomial $f$  ($f$ is a \emph{monogenic polynomial}). In this case, we obtain $R = \gamma R^\vee$, for some $\gamma \in R$ (here, $\gamma$ is a generator of the different ideal), so that $\theta(R^\vee)$ and $\theta(R)$ are related by a linear transformation.  Thus a (dual or non-dual) Ring-LWE problem concerning the lattice $\theta(R)$ or $\theta(R^\vee)$ can be restated as a Poly-LWE problem concerning $P$.

Let $\alpha$ be a root of $f$.  Then $R$ is isomorphic to $P$, via $\alpha \mapsto X$.
An integral basis for $R$ is
        $1, \alpha, \alpha^2, \ldots, \alpha^{n-1}$.
An integral basis for $R^\vee$ is
$        \gamma^{-1}, \gamma^{-1}\alpha, \gamma^{-1}\alpha^2, \ldots, \gamma^{-1}\alpha^{n-1}$.
Let $M_\alpha$ be the matrix whose columns are $\{ \theta(\alpha^i) \}$.
Let $M^\vee_\alpha$ be the matrix whose columns are $\{ \theta(\gamma^{-1}\alpha^i) \}$.
If $\bfv$ is a vector of coefficients representing some $\beta \in K$ in terms of the basis $\{ \alpha^i \}$ for $K/\QQ$, then $\theta(\beta) = M_\alpha \bfv$.  In other words,
$M_\alpha: P \rightarrow \theta(R)$
is an isomorphism (where $P$ is represented as vectors of coefficients).  Similarly,
$M^\vee_\alpha: P \rightarrow \theta(R^\vee)$
is an isomorphism.

\subsection{The spectral norm}

Given an $n \times n$ matrix $M$, its \emph{spectral norm} $\rho = ||M||_2$ is equal to the largest singular value of $M$.  This is also equal to the largest radius of the image of a unit ball under $M$.  This last interpretation allows one to bound the image of a spherical Gaussian distribution of parameter $\sigma$ on the domain of $M$ by another of parameter $\rho \sigma$ on the codomain of $M$ (in the sense that the image of the ball of radius $\sigma$ will map into a ball of radius $\rho \sigma$ after application of $M$).  The spectral norm is bounded above by the Frobenius norm, which is the $\ell_2$-norm on the $n^2$ matrix entries.

The normalized spectral norm of $M$ is defined to be $\rho ' = ||M||_2/\det(M)^{1/n}$.
The condition number of $M$ is $k(M) = ||M||_2||M^{-1}||_2$.

\subsection{Moving the attack from Poly-LWE to Ring-LWE}
\label{sec:movingsub}

Via the isomorphism $M:=M_\alpha^{-1}$ (respectively $M:=(M^\vee_\alpha)^{-1}$), an instance of the non-dual (respectively dual) Ring-LWE problem gives an instance of the Poly-LWE problem in which the error distribution is the image of the error distribution in $\theta(R)$ (respectively $\theta(R^\vee)$).  In general, this may be an elliptic Gaussian distorted by the isomorphism.  If the distortion is not too large, then it may be bounded by a spherical Gaussian which is not too large.  In that case, a solution to the Poly-LWE problem with the new spherical Gaussian error distribution may be possible.  If so, it will yield a solution to the original Ring-LWE problem.

This is essentially the same reduction described in \cite{EHL}. However, those authors assume that the isomorphism is an orthogonal linear map; we are loosening this condition.
The essential question in this loosening is how much the Gaussian distorts under the isomorphism.  Our contribution is an analysis of the particular basis change.

This distortion is governed by the spectral norm $\rho$ of $M$.  If the continuous Gaussian in $\RR^n$ is of parameter $\sigma$ (with respect to the standard basis of $\RR^n$), then the new spherical Gaussian bounding its image is of parameter $\rho\sigma$ with respect to $P$ (in terms of the coefficient representation).  The appropriate analysis for discrete Gaussians is slightly more subtle.
Loosely speaking, we find that a Ring-LWE instance is weak if the following three things occur:

\begin{enumerate}
        \item \label{7monogenic}  $K$ is monogenic.
        \item \label{7root} $f$ satisfies $f(1) \equiv 0 \pmod q$.
        \item \label{7rho} $\rho$ and $\sigma$ are sufficiently small
\end{enumerate}

The first condition guarantees the existence of appropriate isomorphisms to a polynomial ring; the second and third are required for the Poly-LWE attack to apply.
The purpose of the third requirement is that the discrete Gaussian distribution in $\RR^n$ transfers to give vectors $e(x)$ in the polynomial ring having the property that $e(1)$ lies in the range $[-q/4,q/4)$ except with negligible probability; this allows Algorithm \ref{prop:SmallError} and the conclusions of Proposition \ref{prop:SmallError} to apply.

Let us now state our main result.

\begin{theorem}
        \label{thm:redux}
        Let $K$ be a number field such that $K = \QQ(\beta)$, and the ring of integers of $K$ is equal to $\ZZ[\beta]$.  Let $f$ be the minimal polynomial of $\beta$ and suppose $q$ is a prime such that $f$ has root $1$ modulo $q$.
        Finally, suppose that the spectral norm  $\rho$ of $M_\beta^{-1}$ satisfies
        \[
                \rho < \frac{q}{4\sqrt{2\pi}\sigma{n}}.
        \]
        Then the non-dual Ring-LWE decision problem for $K, q, \sigma$ can be solved in time $\widetilde{O}(\ell q)$ with probability $1 - 2^{-\ell}$, using a dataset of $\ell$ samples.
\end{theorem}

\begin{proof}
        Sampling a discrete Gaussian with parameter $\sigma$ results in vectors of norm at most $\sqrt{2\pi}\sigma\sqrt{n}$ except with probability at most $2^{-2n}$ \cite[Lemma 2.8]{LPR2}.  Considering the latter to be negligible, then we can expect error vectors to satisfy $|| \mathbf{v} ||_2 < \sqrt{2\pi}\sigma\sqrt{n}$ and their images in the polynomial ring to satsify
        \[
                |e(1)| = || e(x) ||_1 < \sqrt{n}|| e(x) ||_2 < \sqrt{n} \rho \sqrt{2\pi}\sigma\sqrt{n} = \rho \sqrt{2\pi}\sigma n.
        \]
        Therefore, if
        \[
                \rho \sqrt{2\pi}\sigma n < q/4,
        \]
        then we may apply the attack of Section \ref{sec:SmallError} that assumes $f(1)\equiv 0 \pmod q$ and that error vectors lie in $[-q/4,q/4)$.
\end{proof}

In what follows, we find a family of polynomials satisfying the conditions of the theorem, and give heuristic arguments that such families are in fact very common.  The other cases (other than $\alpha=1$) appear out-of-reach for now, simply because the bounds on $\rho$ are much more difficult to attain.  We will not examine them closely.

\subsection{Choice of $\sigma$}

The parameters of Section \ref{sec:parameters} are used in implementations where the Gaussian is taken over $(\ZZ/q\ZZ)^n$, and security depends upon the proportion of this space included in the `bell,' meaning, it depends upon the ratio $q/\sigma$.  In the case of Poly-LWE, sampling is done on the coefficients, which are effectively living in the space $(\ZZ/q\ZZ)^n$, so this is appropriate.  However, in Ring-LWE, the embedding $\theta(R)$ in $\RR^n$ may be very sparse (i.e. $\theta(R^\vee)$ may be very dense).  Still, the security will hinge upon the proportion of $\theta(R)/q\theta(R)$ that is contained in the bell.  We have not seen a discussion of security parameters for Ring-LWE in the literature, and so we propose that the appropriate meaning of the width of the Gaussian, $w$, in this case is
\begin{equation}
        \label{eqn:s}
        w := \sqrt{2\pi} \sigma' := \sqrt{2\pi} \sigma {\det(M_\alpha)}^{1/n},
\end{equation}
where $\sigma'$ is defined by the above equality.
The reason for this choice is that $\theta(R)$ has covolume $\det(M_\alpha)$; a very sparse lattice (corresponding to large determinant) needs a correspondingly large $\sigma$ so that the same proportion of its vectors lie in the bell.

If $\rho$ represents the spectral norm of $M_\alpha^{-1}$ (which has determinant $\det(M_\alpha)^{-1}$), then
\[
        \rho' := \rho \; {\det(M_\alpha)}^{1/n}
\]
is the normalized spectral norm.  Therefore $\rho/\sigma = \rho'/\sigma'$.  Hence the bound of Theorem \ref{thm:redux} becomes
\begin{equation}
        \label{eqn:rhoprime}
        \rho' < \frac{q}{4 w {n}}.
\end{equation}

\section{Provably weak Ring-LWE number fields}
\label{sec:provable}

        Consider the family of polynomials
        \[
                f_{n,q}(x) =  x^{n} + q-1
        \]
        for $q$ a prime.  These satisfy $f(1) \equiv 0 \pmod q$.  By the Eisenstein criterion, they are irreducible whenever $q-1$ has a prime factor that appears to exponent $1$.  These polynomials have discriminant \cite{Masser} given by
        \[
 (-1)^{\frac{n^2-n}{2}}n^n(q-1)^{n-1}.
        \]

        \begin{proposition}
                \label{prop:family1}
              Let $n$ be power of a prime $\ell$.  If $q-1$ is squarefree and $\ell^2 \nmid ((1-q)^n-(1-q))$ then the polynomials $f_{n,q}$ are monogenic.
        \end{proposition}

        \begin{proof}
                This is a result of Gassert in \cite[Theorem 5.1.4]{Gassert}.  As stated, Theorem 5.1.4 of \cite{Gassert} requires $\ell$ to be an odd prime.  However, for the monogenicity portion of the conclusion, the proof goes through for $p=2$.
        \end{proof}

        \begin{proposition}
                \label{prop:family2}
                Suppose that $f_{n,q}$ is irreducible, and the associated number field has $r_2$ complex embeddings.  Then $r_2 = n/2$ or $(n-1)/2$ (whichever is an integer), and the normalized spectral norm of $M_\alpha^{-1}$ is exactly
                \[
                        2^{-r_2/n}\sqrt{(q-1)^{1-\frac{1}{n}}}.
        \]
        \end{proposition}

        \begin{proof}
                Let $a$ be a positive real $n$-th root of $q-1$.  Then the roots of the polynomial are exactly $a \zeta_{2n}^j$ for $j$ odd such that $1 \le j < 2n$.  The embeddings take $a\zeta_{2n}$ to each of the other roots.  There is $r_1=1$ real embedding if $n$ is odd (otherwise $r_1=0$), and the rest are $r_2$ complex conjugate pairs, so that $n = r_1 + 2r_2$.
                Then the dot product of the $r$-th and $s$-th columns of $M_\alpha^{-1}$ is
                \[
                        \sum_{k=0}^{n-1} a^{r+s}\zeta_{2n}^{(r-s)(2k+1)} = 0
                \]
Therefore, the columns of the matrix are orthogonal to one another.  Hence, the matrix is diagonalizeable, and its eigenvalues are the lengths of its column vectors, which is for the $r$-th column,
\[
\left( \sum_{k=0}^{n-1} || a^r \zeta_{2n}^{2k+1} ||^2 \right)^{1/2} = \sqrt{ n}a^r
\]
Therefore the smallest singular value of $M_\alpha$ is $\sqrt{n}$ and the largest is $\sqrt{n}a^{n-1}$.  Correspondingly, the largest singular values of $M_\alpha^{-1}$ is $1/\sqrt{n}$.

A standard result of number theory relates the determinant of $M_\alpha$ to the discriminant of $K$ via
\[
        \operatorname{det}(M_\alpha) = 2^{-r_2}\sqrt{\operatorname{disc}(f_{n,q})},
\]
where $r_2 \le \frac{n}{2}$ is the number of complex embeddings of $K$.
Combining the smallest singular value with this determinant (the discriminant is given explicitly at the beginning of this section) gives the result.
        \end{proof}

        \begin{theorem}
                \label{thm:redux2}
                Suppose $q$ is prime, $n$ is an integer and $f = f_{n,q}$ satisfies
                \begin{enumerate}
                        \item $n$ is a power of a prime $p$,
                        \item $q-1$ is squarefree,
                        \item $p^2 \nmid ((1-q)^n-(1-q))$,
                        \item we have $\tau > 1$, where
                                \[
                                        \tau := \frac{  q }{2\sqrt{2} w n(q-1)^{\frac12-\frac{1}{2n}}}.
                                \]
                \end{enumerate}
                Then the non-dual Ring-LWE decision problem for $f$ and $w$ (defined by \eqref{eqn:s}) can be solved in time $\widetilde{O}(\ell q)$ with probability $1 - 2^{-\ell}$, using a dataset of $\ell$ samples.
        \end{theorem}

        \begin{proof}
                Under the stated conditions, $f$ has a root $1$ modulo $q$, and therefore Poly-LWE is vulnerable to the attack specified in Algorithm \ref{alg:SmallError}.  The other properties guarantee the applicability of Theorem \ref{thm:redux} via Proposition \ref{prop:family1} and \ref{prop:family2}.
        \end{proof}

        Under the assumption that $q-1$ is infinitely often squarefree, this provides a family of examples which are susceptible to attack (taking, for example, $n$ as an appropriate power of $2$; note that in this case item (3) is automatic).

        Interestingly, their susceptibility increases as $q$ increases relative to $n$.  It is the ratio $\sqrt{q}/n$, rather than their overall size, which controls the vulnerability (at least as long as $q$ is small enough to run a loop through the residues modulo $q$).

        The quantity $\tau$ can be considered a measure of security against this attack ; it should be small to indicate higher security.  For the various parameters indicated in Section \ref{sec:parameters}, the value of $\tau$ is:
\begin{center}
        {\renewcommand{\arraystretch}{1.25}
         \setlength{\tabcolsep}{0.5em}
  \begin{tabular}{| l | l  | l |l |l | l |}
          \hline
          parameters & $P_{LP1}$ & $P_{LP2}$ & $P_{LP3}$ & $P_{GF}$ & $P_{BCNS}$ \\ \hline
          $\tau$ & 0.0136 & 0.0108 & 0.0090 & 0.0063 & 5.0654 \\ \hline
  \end{tabular}}
\end{center}

        The bound on $\tau$ in Theorem \ref{thm:redux2} is stronger than what is required in practice for the attack to succeed.  In particular, the spectral norm of the transformation $M_\alpha^{-1}$ does not accurately reflect the average behaviour; it is worst case.  As $n$ increases, it is increasingly unlikely that error samples happen to lie in just the right direction from the origin to be inflated by the full spectral norm.  Furthermore, we assumed in the analysis of Theorem \ref{thm:redux} an overly generous bound on the error vectors.

        The proof is in the pudding:  in Section \ref{sec:exampleattack} we have implemented several successful attacks.

\section{Heuristics on the prevalence of weak Ring-LWE number fields}
\label{sec:heuristic}

In this section, we argue that many examples satisfying Theorem \ref{thm:redux} are very likely to exist.  In fact, each of the individual conditions is fairly easy to attain.  We will see in what follows that given a random monogenic number field, there is with significant probability at least one prime $q$ for which Ring-LWE is vulnerable (i.e. the bound \eqref{eqn:rhoprime} is attained) for parameters comparable to those of $P_{BNCS}$.  Note that in this parameter range, the spectral norm is expensive to compute directly.

\subsection{Monogenicity}

        Monogenic fields are expected to be quite common in the following sense.  If $f$ of degree $n \ge 4$ is taken to be a random polynomial (i.e. its coefficients are chosen randomly), then it is conjecturally expected that with probability $\gtrsim 0.307$, $P$ will be the ring of integers of a number field \cite{K}.  In particular, if $f$ has squarefree discriminant, this will certainly happen.  Furthermore, cyclotomic fields are monogenic, as are the families described in the last section.

        However, at degrees $n \sim 2^{10}$, the discriminant of $f$ is too large to test for squarefreeness, so testing for monogenicity may not be feasible.  Kedlaya has developed a method for constructing examples of arbitrary degree \cite{K}.

\subsection{Examples, $n = 2^{10}$, $q \sim 2^{32}$}
\label{sec:8examples}

Consider the following examples:
\begin{align*}
        f(x) = x^{1024} + (2^{31}+14)x + 2^{31}, &\quad q = 4294967311, \\
      f(x) = x^{1024} + (2^{31}+2^{30}+22)x + (2^{31}+2^{30}), &\quad q = 6442450967,  \\
     f(x) = x^{1024} + (2^{31}+2^{30}+29)x + (2^{31}+2^{30}+5), &\quad q = 6442450979.
\end{align*}

These examples are discussed at greater length in Section \ref{sec:6examples}, where the method for constructing them is explained.
In each case, $f(1) \equiv 0 \pmod q$.

In this size range, we were not able to compute the spectral norm of $K$ directly in a reasonable amount of time.  In the next few sections we will make persuasive heuristic arguments that it can be expected to have $\rho'$ well within the required bound \eqref{eqn:rhoprime}, i.e. $\rho' < 2^{17}$.  That is, we expect these examples and others like them to be vulnerable.

\subsection{Heuristics for the spectral norm}
\label{sec:hs}

In this section, we will bound the normalized spectral norm by a multiple of the condition number.  Then, for polynomials of the form $x^n + ax+b$, we will argue heuristically that this is frequently small.

Let us recall a standard result of number theory.  For a number field $K$ with $r_1$ real embeddings and $r_2$ conjugate pairs of complex embeddings, the determinant of the canonical embedding is
\begin{equation}
        \label{eqn:detmf}
        \det(M_f) = \sqrt{\Delta_K}2^{-r_2}.
\end{equation}
Under the assumption that $\ZZ[X]/f(X)$ is indeed a ring of integers, $\Delta_K = \Disc(f)$.   
The condition number $k(M_f)$ satisfies
\begin{equation}
        \label{eqn:condtn}
 k(M_f) = ||M_f^{-1}||_2 ||M_f||_2,
\end{equation}
while the normalized spectral norm is
\begin{equation}
        \label{eqn:normrho}
 \rho' = ||M_f^{-1}||_2 \det(M_f)^{1/n}.
\end{equation}
We wish to show that
\[
        \rho' \le 2k(M_f).
\]
To that end, we combine \eqref{eqn:condtn} and \eqref{eqn:normrho}:
\[
        \frac{2k(M_f)}{\rho'} = \frac{2||M_f||_2}{\det(M_f)^{1/n}} \ge 2 \det(M_f)^{1/n} = \Delta_K^{1/2n} 2^{1-r_2/n} \ge 1.
\]

Now restrict to polynomials of the form $f(x) =  x^n + ax + b$.
The condition number of $M_f$ is hard to access theoretically, but heuristically, for random perturbations of any fixed matrix, most perturbations are well-conditioned (having small condition number) \cite{TV}.  The matrix $M_f$ is a perturbation of $M_p$ for $p = x^n+1$.
The extent of this perturbation can be bounded in terms of the coefficients $a$ and $b$, since the perturbation is controlled by the perturbation in the roots of the polynomial.  It is a now-standard result in numerical analysis, due to Wilkinson, that roots may be ill-conditioned in this sense, but the condition number can be bounded in terms of the coefficients $a$ and $b$.  This implies that, heuristically, $k(M_f)$ is likely to be small quite frequently.

In conclusion, we expect to find that many $f(x)$ will have $\rho'$ quite small.

\subsection{Experimental evidence for the spectral norm}

        We only ran experiments in a small range due to limitations of our  Sage implementation (\cite{S}).  The polynomials $x^{32}+ax+b$, $-60 \le a,b \le 60$ were plotted on a $\max\{a,b\}$-by-$\rho'$ plane.  The result is as follows:

        \begin{center}
                \includegraphics[width=4.0in]{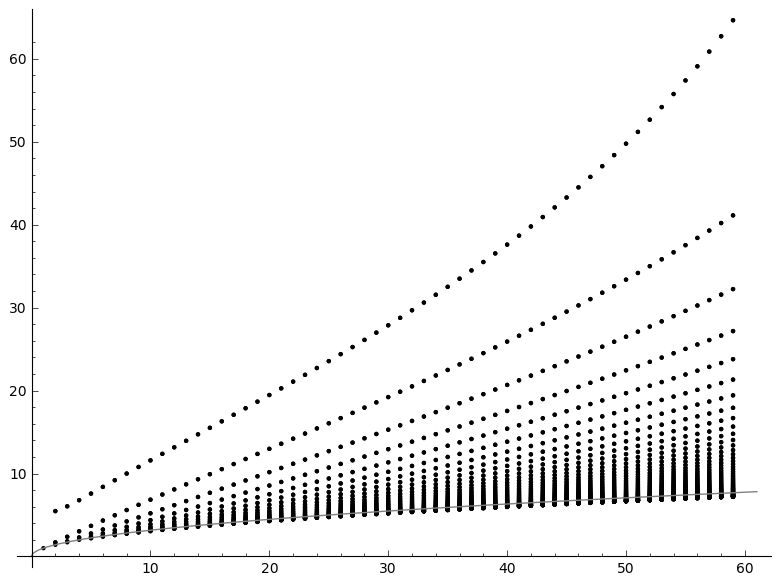}
        \end{center}

        There are some examples with quite high $\rho'$, but the majority cluster low.  The grey line is $y=\sqrt{x}$.  Therefore, we may conjecture based on this experiment, that we may expect to find plenty of $f$ satisfying
$\rho' < \sqrt{\max\{a,b\}}$.

Experimentally, we may guess that the examples of Section \ref{sec:8examples}, for which $n = 2^{10}$ and $\max\{a,b\} \le 2^{30}$, will frequently satisfy $\rho' < 2^{15}$, which is the range required by Theorem \ref{thm:redux}.  (Note that the coefficients cannot be taken smaller if $f$ is to have root $1$ modulo a prime $q \sim 2^{31}$.)

\section{Weak Poly-LWE number fields}
\label{sec:ex1}

\subsection{Finding $f$ and $q$ with roots of small order}
\label{sec:q}

It is relatively easy to generate polynomials $f$ and primes $q$ for which $f$ has a root of given order modulo $q$.  There are two approaches:  given $f$, find suitable $q$; and given $q$, find suitable $f$.  Since there are other conditions one may require for other reasons (particularly on $f$), we focus on the first of these.

Given $f$, in order to find $q$ such that $f$ has a root of small order (this includes the cases $\alpha = \pm 1$), the following algorithm can be applied.

\begin{algorithm}[H]
			\caption{Finding primes $q$ such that $f(x)$ has a root of small order modulo $q$}
			\label{alg:q}
                        \raggedright
{\bf Input:} A non-cyclotomic irreducible polynomial $f(X) \in \ZZ[X]$; and an integer $m \ge 1$.

{\bf Output:} A prime $q$ such that $f(X)$ has a root of order $m$ modulo $q$.

\begin{enumerate}
        \item Let $\Phi_m(X)$ be the cyclotomic polynomial of degree $m$.  Apply the extended Euclidean algorithm to $f(X)$ and $\Phi_m(X)$ over the ring $\QQ[X]$ to obtain $a(X)$, $b(X)$ such that $$a(X)f(X) + b(X)\Phi_m(X) = 1.$$  (Note that $1$ is the GCD of $f(X)$ and $\Phi_m(X)$ by assumption.)
        \item Let $d$ be the least common multiple of all the denominators of the coefficients of $a$ and $b$.
        \item Factor $d$.
        \item Return the largest prime factor of $d$.
\end{enumerate}
\end{algorithm}

It is also possible to generate examples by first choosing $q$ and searching for appropriate $f$.  For example, taking
$f(x) = \Phi_m(x)g(x) + q$
where $g(x)$ is monic of degree $m-n$ suffices.
Both methods can be adapted to find $f$ having any specified root modulo $q$.

\subsection{Examples, $n \sim 2^{10}$, $q \sim 2^{32}$}
\label{sec:6examples}

For the range $n \sim 2^{10}$, we hope to find $q \sim 2^{32}$.  Examples were found by applying Algorithm \ref{alg:q} to polynomials $f(x)$ of the form
$x^n + ax + b$
for $a, b$ chosen from a likely range.  Examples are copious and not difficult to find (see Appendix \ref{A2} for code).

{\bf Case $\alpha=1$.}  A few typical examples of irreducible $f$ with $1$ as a root modulo $q$ are:
\vspace{-0.25em}
\begin{align*}
        f(x) = x^{1024} + (2^{31}+14)x + 2^{31}, &\quad q = 4294967311, \\
      f(x) = x^{1024} + (2^{31}+2^{30}+22)x + (2^{31}+2^{30}), &\quad q = 6442450967,  \\
     f(x) = x^{1024} + (2^{31}+2^{30}+29)x + (2^{31}+2^{30}+5), &\quad q = 6442450979.
\end{align*}
These examples satisfy condition $1$ of Proposition \ref{prop:SmallError} with $\sigma=3$, hence are vulnerable.

{\bf Case $\alpha = -1$.}  Here is an irreducible $f$ with root $-1$:
\begin{align*}
        f(x) = x^{1024} + (2^{31}+9)x - (2^{31}+7), &\quad q = 4294967311 \sim 2^{32}.
\end{align*}
This example similarly satisfies condition $1$ of Proposition \ref{prop:SmallError} and so is vulnerable.

{\bf Case $\alpha$ small order.}  Here is an irreducible $f$ with a root of order $3$:
\begin{align*}
        f(x) = x^{1024} + (2^{16}+2)x - 2^{16}, &\quad q = 1099514773507 \sim 2^{40}.
\end{align*}
This example has $q \sim 2^{40}$; taking this larger $q$ allows us to satisfy \eqref{eqn:SmallSet} of Proposition \ref{prop:SmallSet} and hence it is vulnerable to Algorithm \ref{alg:SmallSet}.

\subsection{Examples of weak Poly-LWE number fields with additional properties}
\label{sec:ex2}

In this section we will give examples of number fields $K = \QQ[x]/(f(x))$ which are vulnerable to our attack on Poly-LWE.  They will be vulnerable by satisfying one of the following two possible conditions:

\vspace{0.25em}

\begin{minipage}[b]{0.90\linewidth}
\begin{description}
        \item[R\label{small}] $f(1) \equiv 0 \; \pmod q$.
        \item[{R$^\prime$}\label{smallprimeprime}] $f$ has a root of small order modulo $q$.
\end{description}
\end{minipage}

\vspace{0.25em}
\noindent
We must also require:

\vspace{0.25em}

\begin{minipage}[b]{0.90\linewidth}
\begin{description}
        \item[Q\label{big}] The prime $q$ can be chosen suitably large.
\end{description}
\end{minipage}

\vspace{0.25em}

The examples we consider are cyclotomic fields and therefore Galois and monogenic. One should note that guaranteeing these two conditions together is nontrivial in general. In addition to these, there are additional conditions for the attack explained in \cite{EHL}.  The desirable conditions are:
\vspace{0.5em}

\begin{minipage}[b]{0.90\linewidth}
\begin{description}
        \item[G\label{galois}] $K$ is Galois.
        \item[M\label{mono}] $K$ is monogenic.
\item[S\label{split}] The ideal $(q)$ splits completely in the ring of integers $R$ of $K$, and $q \nmid [R:\mathbb Z[\beta]]$.
\item[O\label{orthogonal}] The transformation between the canonical embedding of K and the power basis representation of K is given by a scaled orthogonal matrix.
\end{description}
\end{minipage}

\vspace{0.25em}

Conditions \ref{galois} and \ref{split} are needed for the Search-to-Decision reduction and Conditions \ref{mono} and \ref{orthogonal} are needed for the Ring-LWE to Poly-LWE reduction in~\cite{EHL}.

Note that checking the splitting condition for fields of cryptographic size is not computationally feasible in general.  However, we are able to give a sufficient condition for certain splittings which is quite fast to check.

\begin{proposition} \label{easysplit}
Using the notation as above, if $f(2) \equiv 0 \mod q$ then $q$ splits in $R.$
\end{proposition}

\begin{proof}
Since $2^{2^{k-1}} \equiv -1 \pmod q,$ it follows that $(2^\alpha)^{2^{k-1}}\equiv (-1)^\alpha \equiv -1 \pmod q$ for all odd $\alpha$ in $\mathbb Z$. We'll show that $2,2^3, 2^5, \ldots, 2^m$ where $m=2^k-1$ are all distinct mod $q$, hence showing that $f(x)$ has $2^{k-1}$ distinct roots mod $q$ i.e. $f(x)$ splits mod $q$.
Assume that $2^i \equiv 2^j \pmod q$ for some $1\leq i< j \leq 2^k-1$. Then $2^{j-i} \equiv 1 \pmod q$, which means that the order of $2$ modulo $q$ divides $j-i$. However, by the fact below (Lemma \ref{order}),  the order of $2$ mod $q$ is $2^k$, which is a contradiction since $j-i< 2^k.$
\end{proof}

\begin{lemma} \label{order}
Let $q$ be a prime such that $2^{2^{k-1}} \equiv -1 \pmod q$ for some integer $k$. Then the order of $2$ modulo $q$ is $2^k$.
\end{lemma}

\begin{proof}
Let $a$ be the order of $2$ modulo $q$. By assumption $(2^{2^{k-1}})^2 \equiv 2^{2^k} \equiv 1 \pmod q.$ Then $a | 2^k$ i.e. $a=2^\alpha$ for some $\alpha \leq k.$ Say $\alpha \leq k-1.$ Then $1=(2^{2^\alpha})^{2^{k-1-\alpha}}= 2^{2^{k-1}} \equiv -1 \pmod q$, a contradiction.
\end{proof}

The converse of Proposition \ref{easysplit} does not hold. For instance, let $K$ be the splitting field of the polynomial $x^{8}+1$ and $q=401$. Then $q$ splits in $R$. However $f(2)=257 \not\equiv 0 \pmod q$.

We now present a family of examples for which $\alpha=-1$ is a root of $f$ of order two.  Conditions \ref{galois}, \ref{mono}, \ref{split}, \ref{smallprimeprime} (order 2) and \ref{big} are all satisfied.
The field $K$ is the cyclotomic number field of degree $\phi(2^k)=2^{k-1}$, but
instead of the cyclotomic polynomial we take the minimal polynomial of $\zeta_{2^k} + 1.$
In each case, $q$ is obtained by factoring $2^{2^{k-1}}+1$ for various values of $k$ and splitting is verified using Proposition \ref{easysplit}.

\begin{center}
        {\renewcommand{\arraystretch}{1.25}
         \setlength{\tabcolsep}{0.2em}
\begin{tabular}{|c||c|c|c|c|c|c|c|c|c|}
\hline k  & \hspace{0.12cm} 2 \hspace{0.12cm} & \hspace{0.14cm} 3\hspace{0.14cm} & 4 &  5  & 6& 7& 7& 8 & 8  \\
\hline q & 5 & 17 & 257 & 65537$\sim 2^{16}$ &  6700417 $\sim 2^{22}$  & 274177 $\sim 2^{18}$  &    $q_5\sim 2^{45}$ & $q_6\sim 2^{55}$ & $q_1\sim 2^{72}$ \\ \hline
\end{tabular}
\begin{tabular}{|c||c|c|c|c|c|c|c|c| }
\hline k   & 9 & 9 & 10 & 10 & 10& 11& 11& 11  \\
\hline q &  $q_7\sim 2^{50}$   &   $q_2 \sim 2^{205}$  & 2424833$\sim 2^{21}$  & $q_3\sim 2^{162}$  &  $q_4 \sim 2^{328}$ & $q_8\sim 2^{25} $ & $q_9\sim 2^{32}$ & $q_{10} \sim 2^{131}$\\ \hline
\end{tabular}}
\end{center}

\normalsize
Several of these examples are of cryptographic size\footnote{ {$q_1=  5704689200685129054721, q_2= 93461639715357977769163558199606896584051237541638188580280321,$} \\
$ {q_3=7455602825647884208337395736200454918783366342657, q_5= 67280421310721, q_6= 59649589127497217}$\\
${q_4=  741640062627530801524787141901937474059940781097519023905821316144415759504705008092818711693940737}$\\
${q_7=1238926361552897, q_8=45592577, q_9=6487031809, q_{10}=4659775785220018543264560743076778192897 }$}, i.e. the field has degree $2^{10}$ and the prime is of size $\sim 2^{32}$ or greater.  These provide examples which are weak against our Poly-LWE attack, by Proposition \ref{prop:SmallError}.

\section{Cyclotomic (in)vulnerability}
\label{sec:cyc}

One of our principal observations is that the cyclotomic fields, used for Ring-LWE, are uniquely protected against the attacks presented in this paper.  The next proposition states that the polynomial ring of the $m$-th cyclotomic polynomial $\Phi_m$ will never be vulnerable to the attack based on a root of small order.

\begin{proposition} The roots of $\Phi_m$ have order $m$ modulo every split prime $q$.
\end{proposition}

\begin{proof}
Consider the field $\FF_q$, $q$ prime.  Since $\FF_q$ is perfect, the cyclotomic polynomial $\Phi_m(x)$ has $\phi(m)$ roots in an extension of $\FF_q$.  This polynomial has no common factor with $x^k-1$ for $k < m$.  However, it divides $x^m-1$.  Therefore its roots have order dividing $m$, but not less than $m$.  That is, its roots are all of order exactly $m$ in the field in which they live.  Now, if we further assume that $\Phi_m(x)$ splits modulo $q$, then its $\phi(m)$ roots are all elements of order $m$ modulo $q$, so in particular, $m \mid q-1$.  The roots of $\Phi_m(x)$ are all elements of $\ZZ/q\ZZ$ of order exactly $m$.
\end{proof}

The question remains whether there is another polynomial representation for the ring of cyclotomic integers for which $f$ does have a root of small order.  This may in fact be the case, but the error distribution is transformed under the isomorphism to this new basis, so this does not guarantee a weakness in Poly-LWE for $\Phi_m$.

However, it is not necessary to search for all such representations to rule out the possibility that this provides an attack. The ring $R_q \cong \FF_q^n$ has exactly $n = \phi(m)$ homomorphisms to $\ZZ/q\ZZ$.  If $R_q$ can be represented as $(\ZZ/q\ZZ)[X]/f(X)$ with $f(\alpha)=0$, then the map $R_q \rightarrow \ZZ/q\ZZ$ is given by $p \mapsto p(\alpha)$ is one of these $n$ maps.  It suffices to write down these $n$ maps (in terms of any representation!) and verify that the errors map to all of $\ZZ/q\ZZ$ instead of a small subset.  It is a special property of the cyclotomics that these $n$ homomorphisms coincide.  Thus we are reduced to the case above.

\section{Successfully coded attacks}
\label{sec:exampleattack}

The following table documents Ring-LWE and Poly-LWE parameters that were successfully attacked on a Thinkpad X220 laptop with Sage Mathematics Software \cite{S}, together with approximate timings.  For code, see Appendix \ref{sec:code}.  The first row indicates that cryptographic size is attackable in Poly-LWE.  The second row indicates that a generic example attackable by Poly-LWE is also susceptible to Ring-LWE (see Section \ref{sec:heuristic}).  We were unable to test the Ring-LWE attack for $n>256$ only because Sage's built-in Discrete Gaussian Sampler was not capable of initializing (thus we were unable to produce samples to test).  The last two rows illustrate the $\tau$ of Theorem \ref{thm:redux} that is required for security in practice (approximately $\tau < 0.013$ instead of $\tau < 1$ in theory).  Note that these two rings are non-maximal orders ($q-1$ is not squarefree).  In the Ring-LWE rows, parameters were chosen to illustrate the boundary of feasibility for a fixed $n$.  Since the feasibility of the attack depends on the ratio $\sqrt{q}/n$, there is no reason to think larger $n$ are invulnerable (provided $q$ also grows), but we were unable to produce samples to test against.  The Poly-LWE example illustrates that runtime for large $q$ is feasible (runtimes for Poly-LWE and Ring-LWE are the same; it is only the samples which differ).

\begin{center}
        {\renewcommand{\arraystretch}{1.5}
         \setlength{\tabcolsep}{0.5em}
  \begin{tabular}{ c | c | c | c | c | c | c | c}
          {\small case} & $f$ & $q$ & $w$ & $\tau$ & {\small $\substack{\text{samples}\\\text{per run}}$} & {\small $\substack{\text{successful}\\\text{runs}}$} & {\small $\substack{\text{time}\\\text{per run}}$} \\ \hline
    \hline
    {\small Poly-LWE} & $x^{1024}+2^{31}-2$  & $2^{31}-1$ & $3.192$ & N/A & $40$ & $1$ of $1$ & $13.5$ hrs \\
    \hline\hline
    {\small Ring-LWE} & $\substack{x^{128}+524288x \\+524285}$ &  $524287$ & $8.00$ & N/A & $20$ & $8$ of $10$ & $24$ sec \\
    \hline\hline
    {\small Ring-LWE} & $x^{192}+4092$ &  $4093$ & $8.87$ & $0.0136$ & $20$ & $1$ of $10$ & $25$ sec \\
    \hline
    {\small Ring-LWE} & $x^{256}+8190$ &  $8191$ & $8.35$ & $0.0152$ & $20$ & $2$ of $10$ & $44$ sec \\
    \hline
  \end{tabular}}
\end{center}

\appendix
\section{Appendix: Code}
\label{sec:code}
\subsection{Proof of concept for Ring-LWE and Poly-LWE attacks}

The following Sage Mathematical Software~\cite{S} code verifies that Algorithm \ref{alg:SmallError} succeeds on the Poly-LWE and Ring-LWE examples of Section \ref{sec:exampleattack}.  Note that Algorithm \ref{alg:SmallSet} is a minor modification of Algorithm \ref{alg:SmallError}.

This code relies on \texttt{DiscreteGaussianDistributionLatticeSampler}, a built-in package in Sage.  The sampler is incapable of initializing in sufficiently large dimension to fully test the attacks in this paper.   See the related trac ticket \url{http://trac.sagemath.org/ticket/17764}.

Built into the code are several error checks that will be triggered if sufficient precision is not used.

This code is available in electronic form at \url{http://math.colorado.edu/~kstange/scripts.html}.


\scriptsize
\begin{verbatim}
##################################################
# RING-LWE ATTACK #
##################################################

# General preparation of Sage:  Create a polynomial ring and import GaussianSampler, Timer
P.<y> = PolynomialRing(RationalField(), 'y')
from sage.stats.distributions.discrete_gaussian_lattice import DiscreteGaussianDistributionLatticeSampler
RP = RealField(300) # this sets the precision; if it is insufficient, the implementation won't be valid
from sage.doctest.util import Timer

# Give the Minkowski lattice for a given ring determined by a polynomial.
# Also gives a key to which are real embeddings.
def cmatrix(): # returns a matrix, columns basis 1, x, x^2, x^3, ... given in the canonical embedding
    global N, a
    N.<a> = NumberField(f)
    fdeg = f.degree()
    key = [0 for i in range(fdeg)] # 0 = real, 1 = real part of complex emb, 2 = imaginary part
    embs = N.embeddings(CC)
    M = matrix(RP,fdeg,fdeg)
    print "Preparing an embedding matrix:  computing powers of the root."
    apows = [ a^j for j in range(n) ]
    print "Finished computing the powers of the root."
    i = 0
    while i < n:
        em = embs[i]
        if Mod(i,20)==Mod(0,20) or Mod(i,20)==Mod(1,20):
            print "Embedding matrix: ", i, " rows out of ", n, " complete."
        if em(a).imag() == 0:
            key[i] = 0
            for j in range(n):
                M[i,j] = em(apows[j]).real()
            i = i + 1
        else:
            key[i] = 1
            key[i+1] = 2
            for j in range(n):
                M[i,j] = em(apows[j]).real()
                M[i+1,j] = (em(apows[j])*I).real()
            i = i + 2
    return M, key

# Produce a random vector from (Z/qZ)^n
def random_vec(q, dim):
    return vector([ZZ.random_element(0,q) for i in range(dim)])

# Useful function for real numbers modulo q
def modq(r,q):
    s = r/q
    t = r/q - floor(r/q)
    return t*q

# Call sampler
def call_sampler():
    e = sampler().change_ring(RP)
    return e

# Create samples using a lattice (given by latmat and its inverse),
# a Gaussian sampler on that lattice, secret, prime
def get_sample(latmat, latmatinv, sec, qval, keyval):
    e = call_sampler() # create error, in R^n
    dim = latmat.dimensions()[0] # detect dimension of lattice
    pre_a = random_vec(qval, dim) # create a uniformly randomly in terms of basis in cm
    a = latmat*pre_a # create a, in R^n
    b = vecmul_poly(a,sec,latmat,latmatinv) + e # create b, in R^n
    pre_b = latmatinv*b # move to basis in cm in order to reduce mod q
    pre_b_red = vector([modq(c,qval) for c in pre_b])
    b = latmat*pre_b_red
    return [a, b]

# Global choices: setup a field and prime, sampler.
# Set to dummy values that will be altered when an attack is run
q = 1
n = 1
sig = 1/sqrt(2*pi)
Zq = IntegerModRing(q)
R.<x> = PolynomialRing(Zq)
f = y + 1
N.<a> = NumberField(f)
S.<z> = R.quotient(f) # This is P_q
cm,key = cmatrix()
cmi = cm.inverse()
cm
cm53 = cm.change_ring(RealField(10))
cmqq = cm53.change_ring(QQ)
sampler = DiscreteGaussianDistributionLatticeSampler(cmqq.transpose(), sig)

# Set the parameters for the attack
def setup_params(fval,qval,sval):
    global q,n,sig,f,S,x,z,Zq
    f = fval
    n = f.degree()
    q = qval
    Zq = IntegerModRing(q)
    R.<x> = PolynomialRing(Zq)
    sig = sval/sqrt(2*pi)
    S.<z> = R.quotient(f)
    print "Setting up parameters, polynomial = ", f, " and prime = ", q, " and sigma = ", sig
    print "Verifying properties:  "
    print "Prime?", q.is_prime()
    print "Irreducible? ", f.is_irreducible()
    print "Value at 1 modulo q?", Mod(f.subs(y=1),q)
    return True

# Compute the lattices in Minkowski space
def prepare_matrices():
    global cm, key, cmi, cmqq
    print "Preparing matrices."
    cm,key = cmatrix()
    print "Embedding matrix prepared."
    cmi = cm.inverse()
    print "Inverse matrix found."
    cm53 = cm.change_ring(RealField(10))
    cmqq = cm53.change_ring(QQ)
    print "All matrices prepared."
    return True

# Make a vector in R^n into a polynomial, given change of basis matrix and variable to use
def make_poly(a,matchange,var):
    coeffs = matchange*a  #coefficients of the polynomial are given by the change of basis matrix
    pol = 0
    for i in range(n):
        pol = pol + ZZ(round(coeffs[i]))*var^i # var controls where it will live (what poly ring)
    return pol

# Make a polynomial into a vector in Minkowski space
def make_vec(fval,matchange):
    if fval == 0:
        coeffs = [0 for i in range(n)]
    else:
        coeffs = [0 for i in range(n)]
        colist = lift(fval).coefficients()
        for i in range(len(colist)):
            coeffs[i] = ZZ(colist[i])
    return matchange*vector(coeffs)

# Multiplication in the Minkowski space via moving to polynomial ring
def vecmul_poly(u,v,mat,matinv):
    poly_u = make_poly(u,matinv,z)
    poly_v = make_poly(v,matinv,z)
    poly_prod = poly_u*poly_v
    return make_vec(poly_prod,mat)

# Create the sampler on the lattice embedded in R^n
def initiate_sampler():
    global sampler
    print "Initiating Sampler."
    sampler = DiscreteGaussianDistributionLatticeSampler(cmqq.transpose(), sig)
    print "Sampler initiated with sigma", RDF(sig)
    return True

# Produce error vectors, just a test to see how they look
def error_test(num):
    print "Testing the error vector production by producing ", num, " errors."
    errorlist = [sampler().norm().n() for _ in range(num)]
    meannorm = mean(errorlist) # average norm
    maxnorm = max(errorlist) # maximum norm
    print "The average error norm is ", RDF(meannorm/( sqrt(n)*sampler.sigma*sqrt(2*pi) )), " times sqrt(n)*s."
    maxratio = RDF(maxnorm/( sqrt(n)*sampler.sigma*sqrt(2*pi) ))
    print "The maximum error norm is ", maxratio, " times sqrt(n)*s."
    if maxratio > 1:
        print "~~~~~~~~~~~~~~~~~~~~~~~ ERROR ~~~~~~~~~~~~~~~~~~~~~~~~~"
        print "The errors do not satisfy a proven upper bound in norm."
    return True

# Create the secret
secret = 0
def create_secret():
    global secret
    secret = cm*random_vec(q,n)
    return True

# Produce samples
samps = []
numsamps = 1
def create_samples(numsampsval):
    global samps, numsamps
    samps = []
    print "Creating samples"
    for i in range(numsampsval):
        print "Creating sample number ", i
        samp = get_sample(cm, cmi, secret, q, key)
        samps.append(samp)
    numsamps = len(samps)
    print "Done creating ", numsamps, "samples."
    return True

# Function for going down to q
def go_to_q(a,matchange):
    pol = make_poly(a,matchange,x)
    #print "debug got pol:", pol
    pol_eval = pol.subs(x=1)
    #print "debug eval'd to:", pol_eval, " and then ", Zq((pol_eval))
    return Zq(pol_eval)

# Check to make sure moving to q preserves product -- the last two lines should be equal
def sanity_check():
    print "Initiating sanity check"
    mat = cmi
    pvec1 = random_vec(q,n)
    vec1 = cm*pvec1
    pvec2 = random_vec(q,n)
    vec2 = cm*pvec2
    vprod2 = vecmul_poly(vec1,vec2,cm,cmi)
    first_thing = go_to_q(vprod2,mat)
    second_thing = go_to_q(vec1,mat)*go_to_q(vec2,mat)
    if first_thing == second_thing:
        print "Sanity confirmed."
    else:
        print "~~~~~~~~~~~~~~~~~~~~~~~ ERROR ~~~~~~~~~~~~~~~~~~~~~~~~~"
        print "Sanity problem:", first_thing, " is not equal to ", second_thing, "."
        print "Are you sure your ring has root 1 mod q?"
    return True

# Given a list of elements of Z/qZ, make a histogram and zero count
def histoq(data):
    hist = [0 for i in range(10)] # empty histogram
    zeroct=0 # count of zeroes mod q
    for datum in data:
        e = datum
        if e == 0:
            zeroct = zeroct+1
        histbit = floor(ZZ(e)*10/q)
        hist[histbit]=hist[histbit]+1
    return [hist, zeroct]

# Given a list of vectors in R^n, create a histogram of their
# values in Z/qZ under make_poly, together with a zero count
def histo(data,cmi):
    return histoq([go_to_q(datum,cmi) for datum in data])

# Create a histogram of error vectors, transported to polynomial ring
def histogram_of_errors():
    print "Creating a histogram of errors mod q."
    errs = []
    for i in range(80):
        errs.append(sampler())
    hist = histo(errs,cmi)
    print "The number of error vectors that are zero:", hist[1]
    bar_chart(hist[0], width=1).show(figsize=2)
    return True

# Create a histogram of the a's in the samples, transported to polynomial ring
def histogram_of_as():
    print "Creating a histogram of sample a's mod q."
    a_vals = [samp[0] for samp in samps]
    hist = histo(a_vals,cmi)
    print "The number of a's that are zero:", hist[1]
    bar_chart(hist[0], width=1).show(figsize=2)
    return True

# Create a histogram of errors by correct guess
def histogram_of_errors_2():
    print "Creating a histogram of supposed errors if sample is guessed, mod q."
    hist = histoq([ lift(Zq(go_to_q(sample[1],cmi) - go_to_q(sample[0],cmi)*go_to_q(secret,cmi))) for sample in samps])
    print "The number of such that are zero:", hist[1]
    bar_chart(hist[0], width=1).show(figsize=2)
    return True

# Create the secret mod q
lift_s = 0
def secret_mod_q():
    global lift_s
    lift_s = go_to_q(secret,cmi)
    print "Storing the secret mod q.  The secret is ", secret, " which becomes ", lift_s
    return True

# Algorithm 2
# reportrate controls how often it updates the status of the loop; larger = less frequently
# quickflag = True will run only the secret and a few other values to give a quick idea if it works
def alg2(reportrate, quickflag = False):
    print "Beginning algorithm 2."
    numsamps = len(samps)
    a = [ 0 for i in range(numsamps)]
    b = [ 0 for i in range(numsamps)]
    print "Moving samples to F_q."
    for i in range(numsamps):
        sample = samps[i]
        a[i] = go_to_q(sample[0],cmi)
        b[i] = go_to_q(sample[1],cmi)
    possibles = []
    winner = [[],0]
    print "Samples have been moved to F_q."
    for i in range(2):
        if i == 0:
            print "!!!!! ROUND 1: !!!!! First, checking how many samples the secret survives (peeking ahead)."
            iterat = [lift_s]
        if i == 1:
            print "!!!!! ROUND 2: !!!!! Now, running the attack naively."
            possibles = []
            if quickflag:
                print "We are doing it quickly (not a full test)."
                iterat = xrange(1000)
            else:
                iterat = xrange(q)
        for g in iterat:
            if Mod(g,reportrate) == Mod(0,reportrate):
                print "Currently checking residue ", g
            g = Zq(g)
            potential = True
            ctr = 0
            while ctr < numsamps and potential:
                e = abs(lift(Zq(b[ctr]-g*a[ctr])))
                if e > q/4 and e < 3*q/4:
                    potential = False
                    if ctr == winner[1]:
                        winner[0].append(g)
                        print "We have a new tie for longest chain:", g, " has survived ", ctr, " rounds."
                    if ctr > winner[1]:
                        winner = [[g],ctr]
                        print "We have a new longest chain of samples survived:", g, " has survived ", ctr, " rounds."
                ctr = ctr + 1
            if potential == True:
                print "We found a potential secret: ", g
                possibles.append(g)
            if g == lift_s:
                if i == 0:
                    print "The real secret survived ", ctr, "samples."
                #break
    print "Full list of survivors of the ", numsamps, " samples:", possibles
    print "The real secret mod q was: ", lift_s
    if len(possibles) == 1 and possibles[0] == lift_s:
        print "Success!"
        return True
    else:
        print "Failure!"
        return False

# Run a simulation.
def shebang(fval,qval,sval,numsampsval,numtrials,quickflag=False):
    global sig
    print "Welcome to the Ring-LWE Attack."
    n = fval.degree()
    print "The attack should theoretically work if the following quantity is greater than 1."
    print "Quantity: ", RDF( qval/( 2*sqrt(2)*sval*n*(qval-1)^( (n-1)/2/n) ) )
    timer = Timer()
    timer2 = Timer()
    timer.start()
    print "********** PHASE 1: SETTING UP SYSTEM "
    setup_params(fval,qval,sval)
    prepare_matrices()
    print "Computing the adjustment factor for s."
    cembs = (n - len(N.embeddings(RR)))/2
    detscale = RP( ( 2^(-cembs)*sqrt(abs(f.discriminant())) )^(1/n) ) # adjust the sigma,s
    sval = sval*detscale
    sig = sig*detscale
    print "Adjusted s for use with this embedding, result is ", sval
    initiate_sampler()
    print "The sampler has been created with sigma = ", sampler.sigma
    print "Sampled vectors will have expected norm ", RDF(sqrt(n)*sampler.sigma)
    error_test(5)
    print "Time for Phase 1: ", timer.stop()
    timer.start()
    count_successes = 0
    timer2.start()
    for trialnum in range(numtrials):
        print "*~*~*~*~*~*~*~*~*~*~*~*~* TRIAL NUMBER ", trialnum, "*~*~*~*~*~*~*~*~*~*~*~*~*~*~*~*~"
        print "********** PHASE 2: CREATE SECRET AND SAMPLES"
        create_secret()
        create_samples(numsampsval)
        sanity_check()
        print "Time for Phase 2: ", timer.stop()
        timer.start()
        print "********** PHASE 3: HISTOGRAMS"
        histogram_of_errors()
        print "The histogram of errors (above) should be clustered at edges for success."
        histogram_of_as()
        print "The histogram of a's (above) should be fairly uniform."
        histogram_of_errors_2()
        print "The histogram of sample errors (above) should be clustered at edges for success."
        print "Time for Phase 3: ", timer.stop()
        timer.start()
        print "********** PHASE 4: ATTACK ALGORITHM"
        secret_mod_q()
        result = alg2(10000,quickflag)
        print "Result of Algorithm 2:", result
        print "Time for Phase 4: ", timer.stop()
        if result == True:
            count_successes = count_successes + 1
        print "*~*~*~*~*~*~*~*~*~*~*~*~* ", count_successes, " out of ", trialnum+1, " successes so far. *~*~*~*~*~*"
    totaltime = timer2.stop()
    print "Total time for ", trialnum+1, "trials was ", totaltime
    return count_successes

\end{verbatim}
\normalsize

\subsection{Sage code for Algorithm \ref{alg:q}} \label{A2}

The following Sage Mathematics Software~\cite{S} algorithm returns the largest prime $q$ for which a polynomial $f$ has a root of order $m$ modulo $q$.

\scriptsize
\begin{verbatim}
x = PolynomialRing(RationalField(), 'x').gen()
def findq(f,m):
    g = x^m-1
    xg = f.xgcd(g)
    cofs = xg[2].coefficients()
    dens = [ a.denominator() for a in cofs ]
    facs = lcm(dens).factor()
    return max([fac[0] for fac in facs ])
\end{verbatim}
\normalsize

\end{document}